\documentclass[twocolumn,aps,prl,english,superscriptaddress,nobibnotes]{revtex4-2}

\usepackage[colorlinks=true,allcolors=blue]{hyperref}
\usepackage{array}
\usepackage{amsmath}
\usepackage{amssymb}
\usepackage{graphicx}
\usepackage{babel}
\usepackage{mathrsfs}
\usepackage{amsfonts}
\usepackage{epstopdf}
\usepackage{multirow}
\usepackage{color}
\usepackage{bm}
\usepackage{amsthm}
\usepackage{braket}
\usepackage{xcolor}
\usepackage{lipsum}
\usepackage{tikz,fp}
\usepackage{tikz-cd}
\usepackage[export]{adjustbox}
\usepackage{hhline}
\usepackage{tikz}
\usetikzlibrary{quantikz2}
\usetikzlibrary{3d}
\usepackage{tikz-3dplot}
\usepackage{pgfplots}
\usepackage{siunitx}
\usepackage[ruled,linesnumbered,vlined]{algorithm2e}
\usepackage{braket}
\usepackage{tabularray}
\UseTblrLibrary{diagbox}
\usepackage{amsthm}
\newtheorem{definition}{Definition}
\newtheorem{theorem}{Theorem}

\newtheorem{lemma}{Lemma}
\newtheorem{corollary}{Corollary}

\usepackage{CJKutf8} % pdfLaTeX

\newcommand\fd{\bra{\psi_g}\rho\ket{\psi_g}}
\newcommand\emp[1]{\widetilde{\mathbb{E}}[#1]}

\usepackage{bbding}
\usepackage[normalem]{ulem}

\definecolor{darkblue}{rgb}{0.1,0.2,0.6}
\definecolor{crimson}{RGB}{164,16,52}
\definecolor{C0}{RGB}{58,119,175}
\definecolor{C1}{RGB}{239,133,54}

\makeatletter
\def\@fnsymbol#1{%
  \ifcase#1\or
    *\or
    {\text{\Envelope}}\or
    \ddagger\or
    \mathsection\or
    \mathparagraph\or
    \|\or
    **\or
    \dagger\dagger\or
    \ddagger\ddagger
  \else\@ctrerr\fi}
\makeatother

\newif\ifincludesm
\includesmfalse

\begin{document}

\begin{CJK*}{UTF8}{gbsn} % uncomment this if pdfLaTeX
\title{Certifying Quantum States with Uniform Measurements}

\author{Liang Mao (毛亮)}
\thanks{These two authors contributed equally.}
\affiliation{Institute for Advanced Study, Tsinghua University, Beijing, 100084, China}
\affiliation{Institute for Quantum Information and Matter, California Institute of Technology, Pasadena, CA 91125, USA}
\author{Yifei Wang (王逸飞)}
\thanks{These two authors contributed equally.}
\affiliation{Institute for Advanced Study, Tsinghua University, Beijing, 100084, China}
\author{Yingfei Gu (顾颖飞)}
\affiliation{Institute for Advanced Study, Tsinghua University, Beijing, 100084, China}
\author{Chengshu Li (李成疏)}
\email[Corresponding author: ]{chengshu@mail.tsinghua.edu.cn}
\affiliation{Institute for Advanced Study, Tsinghua University, Beijing, 100084, China}
\date{\today}

\begin{abstract}
Qubit-resolved operations and measurements are required for most current quantum information processing schemes. However, these operations can be experimentally costly due to the need for local addressing, demanding significant classical control. A more resource-efficient alternative to extract information is uniform measurement, where a site-independent rotation of qubits is performed before measuring in the computational basis. This operation can be performed in parallel, or globally, in atom- and ion-based platforms, reducing resource cost and increasing fidelity. In this work, we initiate the exploration of the utility of this operation in quantum information processing. In particular, we demonstrate that uniform measurements can certify certain graph states, a family of highly entangled and broadly useful quantum states. We provide a sample-efficient certification algorithm with a proved performance guarantee, together with an experimental scheme based on  analog-mode Rydberg atom arrays. Uniform measurements, therefore, allow direct and efficient characterization of quantum states on quantum platforms in a hitherto unexplored manner. More broadly, our work establishes ``uniformity'' as a meaningful and practically motivated resource rubric for quantum information processing, and offers new insights into the architectural design of quantum computing devices.\\

\end{abstract}

\maketitle
\end{CJK*} % uncomment this if pdfLaTeX

\section{Introduction} 
Recent years have witnessed significant progress in the control of large quantum systems, ranging from large-scale simulation of quantum many-body systems~\cite{Bernien2017,semeghini2021probing,ebadi2021quantum,guo2024site,shao2024antiferromagnetic,iqbal2024non} to the early demonstration of quantum error correction~\cite{bluvstein2024logical,google2025quantum,bluvstein2025architectural}. Leveraging the advantages of controllability, flexibility and tunable interactions, the atom-~\cite{bloch2012quantum,Gross2017,jaksch2000fast,weimer2010rydberg,Saffman2010} and ion-based~\cite{haffner2008quantum,blatt2012quantum,bruzewicz2019trapped} platforms feature compelling application in quantum simulation, as well as serving as  promising candidates for quantum computation. 

However, challenges remain in extracting classical data from these platforms. While performing measurements is vital in quantum information processing, most known protocols are demanding and limited by experimental capabilities. 
When considering many practical quantum information tasks, such as observable estimation~\cite{Huang2020,aaronson2018shadow,huang2022provably,zhao2021fermionic,Wu2026}, state certification or verification~\cite{AGKE15,TM18,GKEA18,dangniamOptimalVerificationStabilizer2020,carrascoTheoreticalExperimentalPerspectives2021,gocaninSampleEfficientDeviceIndependentQuantum2022,HJM24,guptaFewSingleQubitMeasurements2025, Li2026universalefficient} or Hamiltonian learning~\cite{garrison2018does,qi2019determining,bairey2019learning,li2020hamiltonian,anshu2021sample,haah2022optimal}, the state-of-the-art protocols still require randomized measurements~\cite{Elben202203}, demanding qubit-resolved operations.
These operations require significant classical control over individual atoms or ions, as illustrated in Fig.~\ref{fig:intro}(a). On both platforms, one needs to either engineer the control pulse spatial distribution or change the atom positions, hence inevitably increasing resource cost and infidelity.
This motivates us to consider more efficient protocols that take advantage of atomic or ionic platforms on a hardware level.

\begin{figure}[!t]
    \centering
    \includegraphics{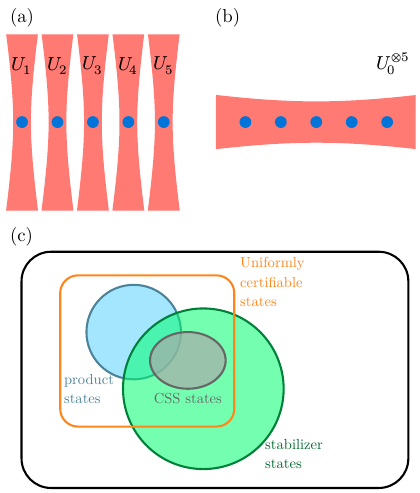}
    \caption{Quantum state measurements and hierarchy of states. (a) General quantum state measurements. Depending on the scheme, one first applies some unitary operations before measuring in the $Z$-basis. (b) Uniform measurements. In this case, we restrict the unitary operations to be uniform. Uniformly certifiable states are those that can be verified from uniform measurements. (c) The relation of uniformly certifiable states with several important types of states. We show that both product states and Calderbank--Shor--Steane (CSS) states are uniformly certifiable. Furthermore, there are interesting stabilizer states beyond the CSS type that are also uniformly certifiable.}
    \label{fig:intro}
\end{figure}

In this paper we propose a novel protocol for quantum state certification, leveraging the experimental feasibility of parallel operations on atomic platforms. Our protocol  utilizes only \textit{uniform measurements}, i.e., single-qubit projection measurements along an arbitrary but qubit-independent direction. Experimentally, this is tantamount to first applying a uniform single-qubit rotation prior to a computational basis measurement, which can be achieved by applying a global laser pulse on all the atoms or ions simultaneously, see Fig.~\ref{fig:intro}(b). 

Despite its immediate experimental relevance~\cite{vankirk2022,Mark2025}, this setup has not been extensively explored. Here we find that uniform measurements are capable of certifying certain quantum states.  Among them, an especially interesting example is a family of \textit{graph states}, which serve as the resource for measurement-based quantum computation~\cite{raussendorf2003measurement,briegel2009measurement}. Preparation and verification of such states have been a central target of current experiments, as demonstrated on platforms including the photon--quantum dot system~\cite{Schwartz2016}, the Rydberg atom array~\cite{Bluvstein2022}, and the photon--atom system~\cite{Thomas2022}. In these and related works, verification typically relies on stabilizer measurements that are technically demanding. Our protocol replaces them with uniform global operations, providing an immediately friendlier alternative to these platforms.

In the following, we first characterize the capability of uniform measurements. After a few illuminating simple examples, we discuss how to certify certain graph states. We provide a sample-efficient certification algorithm with a proved performance guarantee, using only uniform measurements. We also propose an experimental scheme based on the Rydberg atom array operated in analog mode. We conclude with a discussion of the broader architectural relevance and possible further extensions.

\section{Quantum state certification} 
Given a quantum state ${\rho}$ synthesized in the lab, quantum state certification is the task of certifying whether the fidelity $\bra{\psi}{\rho}\ket{\psi}$ to a target state $\ket{\psi}$ is sufficiently close to one. For a generic $\ket{\psi}$, the projector $\ket{\psi}\bra{\psi}$ is a highly complicated non-local operator, which makes it difficult to directly measure its expectation value. In modern certification protocols designed for a fully-controllable quantum computer,  global unitary operations or qubit-resolved random measurements are usually necessary~\cite{AGKE15,TM18,GKEA18,dangniamOptimalVerificationStabilizer2020,carrascoTheoreticalExperimentalPerspectives2021,gocaninSampleEfficientDeviceIndependentQuantum2022,HJM24,guptaFewSingleQubitMeasurements2025, Li2026universalefficient}.

In this paper, we deal with a special scenario: $\ket{\psi}$ is the ground state of some sparse Hamiltonian with a known energy~\cite{jones2005parts,huber2016characterizing,karuvade2019uniquely,yu2023learning}, such as the ground states of physically local Hamiltonians or the stabilizer states. For these states, the parent Hamiltonian ${H}$ is a sum of local (or low complexity) terms ${H}=\sum_i{h}_i$. Then $\ket{\psi}$ can be efficiently certified by measuring each ${h}_i$ and comparing the results to the ground state energy. Note that despite its simplicity, estimating ${h}_i$ still typically requires qubit-resolved operations and measurements \cite{cotler2020quantum,bonet2020nearly,garcia2020pairwise}. We will see how uniform measurements help to certify a class of these states in the following.

\section{Uniform measurements}
Before demonstrating the utility of uniform measurements for certification, we now characterize their capacity. 
\begin{definition}
    [Uniform measurements]
    We call a quantum measurement a uniform measurement if it applies a global parallel single-qubit rotation $\bigotimes_{i=1}^n U$ before doing the computational basis measurement.
\end{definition}

For the sake of illustration,  we first restrict ourselves to uniform measurements of qubits in the $x-z$ plane. Generalization to an arbitrary direction on the Bloch sphere is straightforward.
Uniform measurement in the $x-z$ plane can be performed by measuring in the $Z$ basis after a site-independent rotation in the $x-z$ plane $U=\prod_{i=1}^N \exp(-\mathrm{i}Y_i\theta/2)$ on all the $N$ qubits,
and reading measurement outputs from any subset $I\subseteq \{1,...,N\}$. This gives the expectation value
\begin{equation}
\left\langle \bigotimes_{i\in I}(\cos\theta Z_i+\sin\theta X_i)\right\rangle,
\end{equation}
Expanding the expression, we get a linear combination of expectation values of Pauli strings, each of which is a symmetrized summation of a product of $\alpha$ $Z$ operators and $(|I|-\alpha)$ $X$ operators,
\begin{equation}\label{sym}
E(\alpha;I) = \left\langle \sum_{\substack{J\subseteq I\\|J|=\alpha}}\bigotimes_{\substack{j\in J\\k\in I\backslash J}} Z_j X_k\right\rangle,\quad \alpha=0,...,|I|,
\end{equation}
weighted by a factor $\cos^{\alpha}\theta\sin^{|I|-\alpha}\theta$.
Then by repeating with $|I|+1$ different $\theta$'s,
one can obtain all the expectation values $E(\alpha;I)$.
Taking $I=\{1,2,3\}$ as an example, what can be measured is $E(0;I) = \langle X_1X_2X_3\rangle, E(1,I) = \langle X_1X_2Z_3+X_1Z_2X_3+Z_1X_2X_3\rangle, E(2;I) = \langle X_1Z_2Z_3+Z_1X_2Z_3+Z_1Z_2X_3\rangle$, and $E(3;I)=\langle Z_1Z_2Z_3\rangle$.

Taking into account all three directions $x,y,z$, one can estimate the expectation values of all the symmetrized operators with uniform measurements. 
\begin{definition}
    [Symmetrized operators]
    An operator $O$ is called a symmetrized operator if on its support $S\subseteq[n]$, $O$ commutes with all the symmetric group elements $\pi\in\mathcal{S}_{|S|}$.
\end{definition}\noindent
For any size-$k$ subset $I$ of $N$ qubits, the number of independent symmetric operators is $\binom{k+2}{2}$.
So the total number of independent operators one can evaluate is
\begin{equation}
\begin{split}
\mathcal{N}(N)&{}=\sum_{k=1}^{N}\binom{N}{k}\binom{k+2}{2}\\
&{}=2^{N-3}(N^2+7N+8)-1.   
\end{split}
\end{equation}
Note that $\mathcal{N}(N)<4^N-1$ except for $N=1$. So we do not expect enough information for a full tomography of a generic density matrix. However, we will see that much can be learned from only uniform measurements.

\section{Elementary examples}

\textbf{Example 1}. Any direct product state can be certified with uniform measurements.  For a product state $\ket{\psi}=\otimes_{i=1}^N\ket{\psi_i}$, we only need to estimate the fidelity for every $\ket{\psi_i}$. Because $\ket{\psi_i}\bra{\psi_i}$ can be expanded with Pauli operators, $\bra{\psi_i}{\rho}\ket{\psi_i}$ can be estimated by uniform measurements through $x,y,z$ directions and analyzing the results for each qubit separately.

Since we will mostly focus on stabilizer states below, we wish to point out first that uniform measurements have capability beyond such states. More generally, we have the following result.

\textbf{Proposition 1.} If two states $\psi_1,\psi_2$ defined on different qubits can each be uniformly certified, then so can the product state $\psi_1\otimes\psi_2$. This follows directly from the definition.

\textbf{Example 2.} Calderbank--Shor--Steane (CSS) states~\cite{Calderbank1996,Steane1996}. These states are stabilizer states where the generators can be chosen to be products of either only $Z$ or only $X$ operators. 
Such states lie at the heart of current quantum error correction schemes.
The CSS states can be certified by first performing uniform measurements along $x$ and $z$ directions, and then estimating each stabilizer.
This example shows that uniform measurements are also capable of capturing a rich entanglement structure. 

\begin{figure}
    \centering
    \includegraphics[]{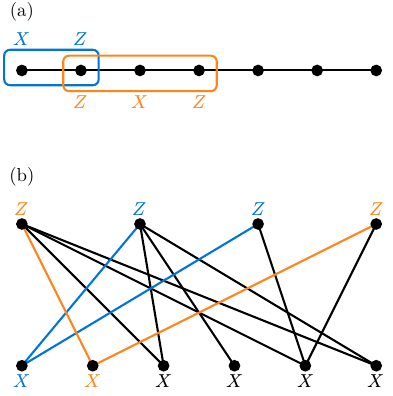}
    \caption{Illustration for graph states. 
    (a) A graph state on a 1D chain, also known as the 1D cluster state. We mark the stabilizer on the left boundary $X_1Z_2$ in blue and a bulk stabilizer $Z_2X_3Z_4$ in orange.
    (b) A subclass of bipartite graph states, where every vertex on the upper side has an even number of adjacent edges. We mark two stabilizers in blue and orange respectively. The product of all $X$ operators on the lower side is in the stabilizer group. 
    }
    \label{illustration}
\end{figure}

\section{Graph states} 
In the above examples, the stabilizers of the states all have a symmetric structure on their supports. Surprisingly, the certification capability of uniform measurements goes beyond stabilizer states with only symmetric stabilizers. We find that all bipartite graph states with even degrees on one partition can be certified with uniform measurements. As an illustration, we demonstrate this with an example of one-dimensional (1D) graph state below.

The 1D graph state $\ket{\psi_g}$ is depicted in Fig.~\ref{illustration}(a) For an $N$-qubit chain, $\ket{\psi_g}$ is stabilized by $Z_{i-1}X_iZ_{i+1}$ for $2\leq i\leq N-1$, together with $X_1Z_2$ and $Z_{N-1}X_N$ on the boundaries. Here we use $Z_i$ and $X_i$ to denote Pauli operators on the $i$th qubit. $N$ is chosen to be an odd number, so that any qubit $i\in\text{even}$ has an even degree. Note that the one-dimensional graph state without boundary stabilizers serves as a symmetry-protected topological phase~\cite{Son2012} and a resource state for measurement-based quantum computation~\cite{Raussendorf2001}, and experimental realizations thereof have been actively pursued~\cite{Larsen2019,Cao2023,Ferreira2024}. Our certification protocol easily generalizes to it. Here we keep the boundary terms for the sake of generality.

Assume that the synthesized state in the lab matches exactly the 1D graph state, ${\rho}=\ket{\psi_g}\bra{\psi_g}$. One can verify it through the following procedures. First, note that ${\rho}$ has a $\mathbb{Z}_2$ symmetry generated by $U_X=\prod_{i\in\text{odd}}X_i$, whose existence can be verified by uniform measurements. Next, one estimates the expectation values of bulk stabilizers by measuring 
\begin{align}\label{stabilizer1}
    M_i=Z_{i-1}X_iZ_{i+1}+Z_{i-1}Z_iX_{i+1}+X_{i-1}Z_iZ_{i+1}
\end{align}
for $2\leq i\leq N-1$, and boundary stabilizers by measuring
\begin{align}\label{stabilizer2}
    M_1=X_1Z_2+Z_1X_2,\quad M_N=X_{N-1}Z_N+Z_{N-1}X_N,
\end{align}
through uniform measurements (Eq.~\eqref{sym}). $Z_{i-1}Z_iX_{i+1}$, $X_{i-1}Z_{i}Z_{i+1}$, $Z_1X_2$ and $X_{N-1}Z_N$ all anticommute with $U_X$, so their expectation values should all be zero, according to the following lemma.
\begin{lemma}\label{lm1}
    Let $\mathcal{P}_1,\mathcal{P}_2\in \{I,X,Y,Z\}^{\otimes N}$ be two Pauli strings and $\{\mathcal{P}_1,\mathcal{P}_2\}=0$. Then their expectation values $\braket{\mathcal{P}_i}=\operatorname{Tr}({\rho}{\mathcal{P}}_i)$ with any quantum state ${\rho}$ should satisfy $\braket{\mathcal{P}_1}^2+\braket{\mathcal{P}_2}^2\leq1$.
\end{lemma}\noindent
For a single-qubit state, $\braket{X}^2+\braket{Z}^2\leq1$ always holds. A general proof can be found in the Supplementary Materials. From the lemma, the expectation values of stabilizers in Eqs.~\ref{stabilizer1} and \ref{stabilizer2} should all be exactly one. This verifies ${\rho}=\ket{\psi_g}\bra{\psi_g}$.

\begin{algorithm}[t!]
	\caption{Certifying the 1D graph state}
	\label{1}
    \SetKwInOut{Input}{Input}    
\SetKwInOut{Output}{Output} 
\Input{	$3T$ samples of unknown quantum state $\rho$.}
\Output{\textsc{Certified} or \textsc{Failed} depending on the fidelity $\braket{\psi_g|\rho|\psi_g}$ (See Theorem \ref{thm: main}).}

\BlankLine

 Take uniform measurements along the basis $X^{\otimes N}$, $\big(\frac{X+Z}{\sqrt{2}}\big)^{\otimes N}$, and $\big(\frac{X-Z}{\sqrt{2}}\big)^{\otimes N}$ for $T$ times each. Collect the measurement outcomes $b_x^{i,t},b_{xz+}^{i,t},b_{xz-}^{i,t}=\pm 1$ of each measurement event of each qubit, for qubits $i=1,2,\cdots N$ and measurement events $t=1,2,\cdots T$.

 Estimate $\braket{U_X}$ by the empirical average $\emp{U_X}=\frac{1}{T}\sum_{t=1}^T \prod_{i\in\text{odd}}b_x^{i,t}$. If $\emp{U_X}<1-13\epsilon/4$, output \textsc{Failed}.

 Estimate $\braket{M_i}$ by the empirical average
\begin{align*}
    \emp{M_i}&=\frac{1}{T}\sum_{t=1}^T\Big[\sqrt{2}\big(b_{xz+}^{i-1,t}b_{xz+}^{i,t}b_{xz+}^{i+1,t} 
    +\notag\\
    &\quad b_{xz-}^{i-1,t}b_{xz-}^{i,t}b_{xz-}^{i+1,t} \big)
    -b_x^{i-1,t}b_x^{i,t}b_x^{i+1,t}\Big]
\end{align*}
for $2\leq i\leq N-1$, and 
\begin{align*}
    \emp{M_1}&=\frac{1}{T}\sum_{t=1}^T\Big[b_{xz+}^{1,t}b_{xz+}^{2,t}-b_{xz-}^{1,t}b_{xz-}^{2,t} \Big]\notag\\
    \emp{M_N}&=\frac{1}{T}\sum_{t=1}^T\Big[b_{xz+}^{N-1,t}b_{xz+}^{N,t}-b_{xz-}^{N-1,t}b_{xz-}^{N,t} \Big]
\end{align*}
If $\emp{M_i}> 1-9\sqrt{\epsilon}$ for all $1\leq i\leq n$, output \textsc{Certified}. Otherwise output \textsc{Failed}.
% \end{algorithmic}
\end{algorithm}

More generally, when the fidelity $\bra{\psi_g}{\rho}\ket{\psi_g}$ is less than one, the estimated values of Eqs.~(\ref{stabilizer1}, \ref{stabilizer2}) can be bounded with Lemma \ref{lm1}.
Building on the analysis, we develop a general certification algorithm Algo.~\ref{1} for the one-dimensional graph state that works for finite fidelity, using only uniform measurements. The algorithm has a proved performance guarantee: 
\begin{theorem}[]\label{thm: main}
	For any $\epsilon<1/(64N^2)$, Algorithm \ref{1} outputs \textsc{Certified} with probability at least $2/3$ when $\fd>1-\epsilon$, and outputs \textsc{Failed} with probability at least $2/3$ when $\fd<1-8N\sqrt{\epsilon}$, using a number of samples at most $T=\mathcal{O}\big(1/\epsilon^2\big)$.
\end{theorem}\noindent
The proof can be found in the Supplementary Materials.

This certification algorithm works far beyond the 1D systems. A general graph state on a graph $G=\{V,E\}$ is defined to be stabilized by $X_{v}\prod_{(v,v')\in E}Z_{v'}$ for all $v\in V$, as illustrated in Fig.~\ref{illustration}(b).  Consider $G$ to be a bipartite graph with partitions $A$ and $B$, and assume the degrees of $v\in A$ are always even. Then $U_X=\prod_{v\in B}X_v$ is a symmetry of the graph state. Define $V'(v)=\{v'|(v,v')\in E\}\cup\{v\}$ to be the neighbors of $v$ and itself. The expectation values of all the operators
\begin{align}
    \sum_{v'\in V'(v)}\Big(X_{v'}\prod_{v''\in V'(v)\smallsetminus\{v'\}}Z_{v''}\Big),\quad v\in V
\end{align}
can be bounded using $U_X$ and Lemma \ref{lm1}. So Algo.~\ref{1} generalizes directly to these cases. We note that in general if the stabilizer term is supported on $L$ qubits, it generally requires measurement of $L$ distinct symmetrized operators. However, in most physically relevant settings, including 1D and 2D lattice graphs, $L$ is a small constant independent of system size. Even for more complex constructions such as expander graphs, the degree (and hence $L$) remains bounded by a constant. So the $L$-dependence does not affect the efficiency of our algorithm.

\begin{figure}
    \centering
    \includegraphics[width=0.9\linewidth]{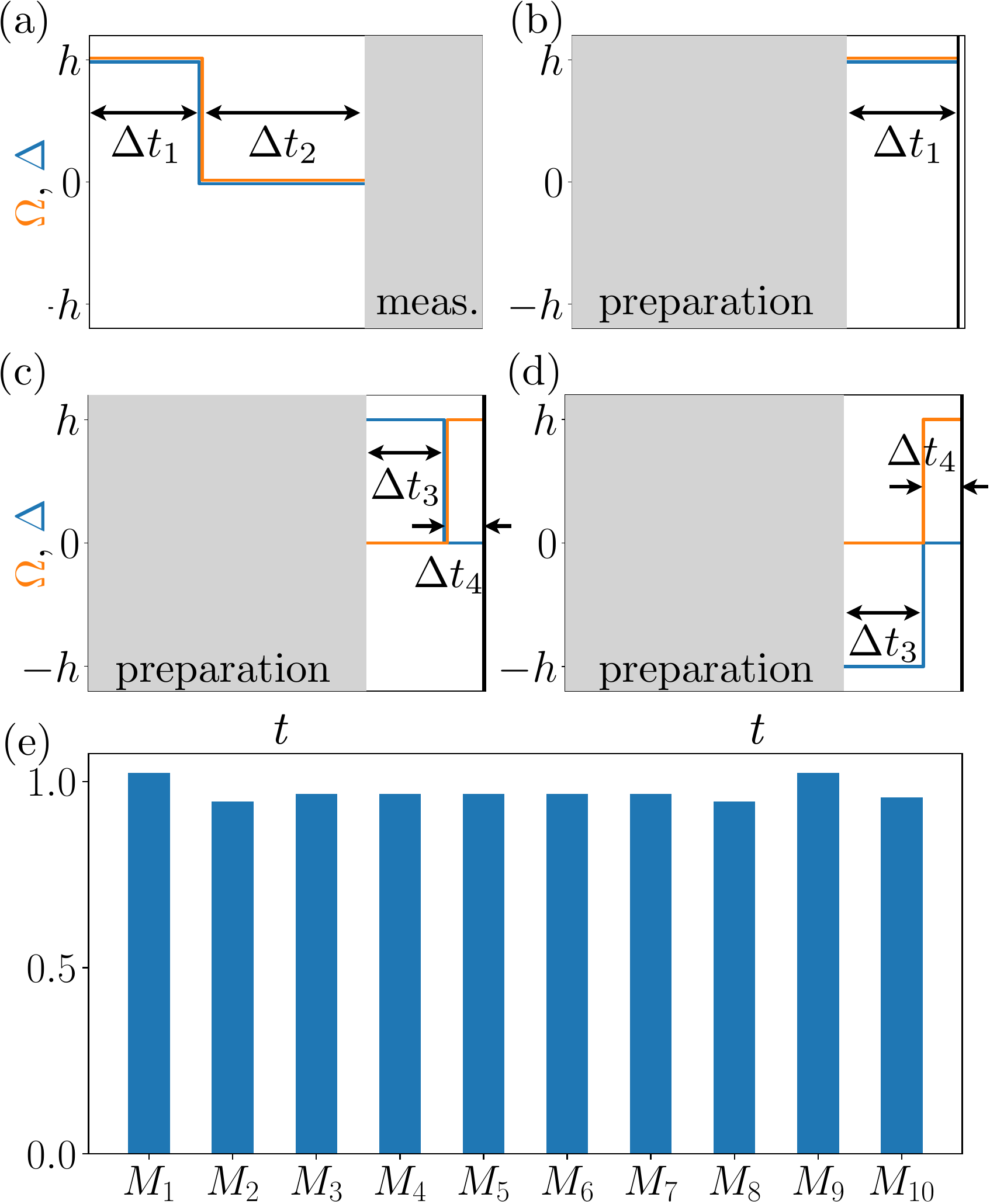}
    \caption{Experimental scheme in a Rydberg atom chain with an odd number of sites and open boundary conditions. In the numerical simulation we take $N=9$. (a--d) The time sequence of preparing and measuring the graph state. The measurement is performed along the (b) $x$, (c) $(x+z)/\sqrt{2}$, and (d) $(x-z)/\sqrt{2}$ directions. Blue and orange lines denote the detuning $\Delta(t)$ and the Rabi frequency $\Omega(t)$ respectively. The time intervals are $\Delta t_1=1/(2\sqrt{2}h)$, $\Delta t_2=1/2$, $\Delta t_3=1/(4h)$, $\Delta t_4=1/(8h)$. The vertical black line denotes the moment of measurement. (e) The measurement results for $M_i,i=1,\ldots,10$. We take $h=20C_6$ in all calculations.}
    \label{fig:rydberg}
\end{figure}

\section{Experimental scheme in a Rydberg atom array}
Now we propose an experimental scheme to prepare and certify the 1D graph state, using our certification protocol. We first discuss a unitary evolution that prepares the 1D graph state. Previously, such states have been prepared in the quantum computation modes~\cite{Bluvstein2022}, namely within
two hyperfine levels, in the Rydberg atom array where Rydberg interaction is used to perform a two-qubit gate~\cite{Saffman2010}. The current protocol, however, works exclusively in the quantum simulation mode.

The Hamiltonian of the system reads
\begin{equation}
\begin{split}
H/2\pi={}&\frac{\Omega(t)}{2}\sum_i(|r_i\rangle\langle g_i|+|g_i\rangle\langle r_i|) - \Delta(t)\sum_i n_i\\
&{}+\sum_{i<j}\frac{C_6}{|i-j|^6}n_i n_j,
\end{split}\end{equation}
where $|g,r\rangle$ denote the ground and Rydberg states respectively. The population of the ground state atoms can be measured by imaging the system after ionizing the Rydberg atoms. The two states serve as the spin degree of freedom of the graph state. $n=|r\rangle\langle r|=(1-Z)/2$ is the Rydberg state number operator. $\Omega(t)$ and $\Delta(t)$ are the Rabi frequency and detuning, and $C_6$ measures the strength of the van der Waals interaction. The system is defined on an $N$-site open chain with $N$ being odd. We take $C_6=1$ as the energy scale. Starting from the initial state $\ket{g}^{\otimes N}$ , we first rotate by an angle of $\pi$ along the $(1,0,1)$-direction, preparing the state $\ket{+}^{\otimes N}=[(\ket{g}+\ket{r})/2]^{\otimes N}$. In terms of Rydberg operators, we take $\Omega=\Delta=h$ for a large value $h$ and hold for $t=1/(2\sqrt{2}h)$. For a realistic Rydberg system where $h$ is a few MHz, we need to choose the nearest-neighbor interaction to be around a hundred kHz, corresponding to a lattice constant of 7\si{\micro m} for two Rb atoms excited to the 50$s$ state. Then we turn off $\Delta$ and $\Omega$, and the Rydberg interaction entangles the state, leading to the 1D graph state. Indeed, the Rydberg interaction results in a nearest neighbor C$Z$ gate when applied for a time interval $t=1/2$, and $\mathrm {C}Z_{ij} X_j \mathrm{C}Z_{ij}=Z_i X_j$. The time sequence of this preparation scheme is shown in Fig.~\ref{fig:rydberg}(a).

Once the 1D graph state is prepared, we can perform Algo.~\ref{1} to certify it. The protocol involves collecting $N+1$ expectation values, namely $\braket{U_X}=\braket{\prod_{i\in\text{odd}}X_i}$ and that of Eqs.~(\ref{stabilizer1}, \ref{stabilizer2}). In order to estimate these expectation values, all we need is to measure in three different directions, namely $x$, $(x+z)/\sqrt{2}$, and $(x-z)/\sqrt{2}$. They require a further rotation before measuring each atom, as shown in Fig.~\ref{fig:rydberg}(b--d). After a straightforward calculation, the $N+1$ measurement results can be obtained. 

Note that while general uniform measurements of a weight-$|I|$ cluster require $|I|+1$ independent measurements, only $|I|=3$ are needed in the above example. The reason is that this combination ($XZZ+ZXZ+ZZX$) only involves a fixed odd number of $X$'s (one in this case) in each term. All such combinations can be measured by one less than the general case. To see this, we consider a general $|I|$ and measure in angles of (a) $\theta=n\pi/|I|$ and $\theta=\pi-n\pi/|I|$ in pairs ($n\in[\lfloor|I|/2\rfloor]$, (b) $\theta=0$ for all $|I|$, and (c) $\theta=\pi/2$ for odd $|I|$. These are $|I|+1$ independent measurements. Now, to obtain any combination involving a fixed odd number of $X$'s, one only needs (a) or (a, c) depending on the parity of $|I|$. This is because there are $\lfloor(|I|+1)/2\rfloor$ such combinations and hence we need the same number of equations. The measurements in (a) result in $\lfloor|I|/2\rfloor$ equations, which suffice for even $|I|$ and need one more equation for odd $|I|$, both giving the total count of $|I|$. 

We numerically illustrate the validity of our experimental scheme through a $N=9$ system, where the results can be found
 in Fig.~\ref{fig:rydberg}(e). Here we label $M_1=\langle X_1 Z_2+Z_1 X_2\rangle$, $M_9=\langle Z_{8}X_{9}+X_{8}Z_9\rangle$ and $M_i=\braket{Z_{i-1}X_iZ_{i+1}+X_{x-1}Z_iZ_{i+1}+Z_{i-1}Z_iX_{i+1}}$ for $2\leq i\leq 8$. $M_{10}$ denotes $\braket{U_X}$.
 All values are close to unity, as expected. The errors come from $h$ being finite and Rydberg interactions beyond the nearest neighbor.  

\section{Discussions and conclusions}

In this work, motivated by an attempt to alleviate classical control, we propose to use uniform measurements to certify quantum states. We characterize the capacity of uniform measurements, and illustrate its rich behavior through several examples. In particular, for the non-CSS case of the graph state, we propose a rigorous certification algorithm as well as an experimental protocol based on the analog-mode Rydberg atom array. This proposal relies on standard techniques on the platform and is very much within current experimental capacities. 

While we present the graph states as a concrete example, uniformly certifiable states may be much more general. However, finding all the uniformly certifiable states seems to be a challenging task, for they lack a good mathematical structure. We leave this exciting and useful question for further explorations.

Our work offers new insight into the architectural design of quantum computing devices. Modular and distributed quantum architectures are now being demonstrated~\cite{bluvstein2025architectural}. In such settings, different zones of a device may be tasked with state preparation (“manufacturing”), computation (“operating”), and quality control (“verifying”). Our scheme naturally supports this division: the verifying zone can operate without local addressing, enabling the reallocation of scarce resources (e.g., optical components, lasers) to the operating zone and improving overall resource-efficiency.

Broadly speaking, the complexity of quantum computation involves multiple rubrics, with the number of qubits and gates being the most basic~\cite{Watrous2009}. Upon more careful inspection, depending on the nature of the quantum platform, one also expects certain operations to be more resource-efficient than others. Generally, we expect that uniform measurements will lead to a significant improvement of resource cost across various platforms. This might be quantified by a more nuanced definition of complexity, and can be used to systematically improve the development of quantum subroutines. The current work initiates an early step in this direction.

An important practical aspect of quantum state verification is error localization. In our framework, uniform measurements restrict accessible observables to symmetrized operators, which limits spatial resolution. As a result, the protocol focuses on certifying global properties such as fidelity, rather than pinpointing errors at the single-qubit level. Nevertheless, coarse-grained diagnostic information can still be obtained. For graph states, the measured quantities ($M_i$) are associated with local stabilizer structures. The deviations from their ideal values can indicate the presence of errors in corresponding regions. Achieving finer error localization would likely require going beyond strictly uniform measurements, for example by incorporating limited local control or combining multiple measurement settings. Exploring such extensions is an interesting direction for future work.\\

\noindent
{\bf Acknowledgments}
We thank Hsin-Yuan Huang, Yadong Wu, Pengfei Zhang, Yihui Quek, Yi-Zhuang You, Haimeng Zhao, Hong-Ye Hu, and Huangjun Zhu for helpful discussions. \\

\noindent
{\bf Funding}
This work is supported by Quantum Science and Technology-National Science and Technology Major Project under Grant No.~2025ZD0300400, National Natural Science Foundation of China Grant No.~12504307 (C.L.) and No. 12575022 (Y.G.), the National Key R\&D Program of China under Grant No.~2023YFA1406702, the Scientific Research Innovation Capability Support Project for Young Faculty under Grant No.~SRICSPYF-ZY2025157, the Shanghai Committee of Science and Technology under Grant No.~25LZ2600800, and Tsinghua University Dushi program (C.L., Y.G.). \\

\noindent
{\bf Availability of data and materials}
The codes and data for our numerical calculations are available at \url{https://github.com/chengshul/Uniform_measurement}.\\

% \noindent
% {\bf Author contributions}
% C.L. conceived the idea and supervised the research. All authors contributed to the theoretical formulation. L.M. and Y.W. contributed to the detailed analysis. C.L. wrote the manuscript. All authors have read and approved the manuscript.\\

% \noindent
% {\bf Competing interests}
% The authors declare no competing interests.\\

\bibliography{biblio}

\onecolumngrid

\clearpage
\subsection*{\large Supplementary Materials for \\ ``Certifying Quantum States with Uniform Measurements''}
\normalsize
\setcounter{equation}{0}
\setcounter{figure}{0}
\setcounter{table}{0}
\setcounter{section}{0}
\setcounter{page}{1}
\renewcommand{\theequation}{S\arabic{equation}}
\renewcommand{\thefigure}{S\arabic{figure}}
\renewcommand{\thetable}{S\arabic{table}}
\renewcommand{\bibnumfmt}[1]{[S#1]}

\makeatletter

\twocolumngrid

\section{Rigorous proof of algorithm performance}
\setcounter{lemma}{0}
\setcounter{theorem}{0}
In this section, we provide a rigorous proof of the performance of our certification algorithm. The proof replies on several technical prerequisites.

First, we show that the expectation values of two anticommuting Pauli strings are bounded
\begin{lemma}\label{lm1}
    Let $\mathcal{P}_1,\mathcal{P}_2\in \{I,X,Y,Z\}^{\otimes N}$ be two Pauli strings and $\{\mathcal{P}_1,\mathcal{P}_2\}=0$. Then their expectation values $\braket{\mathcal{P}_i}=\operatorname{Tr}({\rho}{\mathcal{P}}_i)$ with any quantum state ${\rho}$ should satisfy $\braket{\mathcal{P}_1}^2+\braket{\mathcal{P}_2}^2\leq1$.
\end{lemma}
\begin{proof}
    For any real number $\alpha$, define an Hermitian operator $A_\alpha = \mathcal{P}_1 + \alpha \mathcal{P}_2$.
    The fluctuation of $A_\alpha$ must be non-negative, that is
    $\braket{(A_\alpha - \braket{A_\alpha})^2} \geq 0$ for any $\alpha\in \mathbb{R}$.
    Expanding the inequality we have
    \begin{equation}
        \alpha^2\left(1-\braket{\mathcal{P}_2}^2\right) - 2\alpha\braket{\mathcal{P}_1}\braket{\mathcal{P}_2} + \left(1-\braket{\mathcal{P}_1}^2\right) \geq 0
    \end{equation}
    for any $\alpha\in \mathbb{R}$.
    The condition for this to hold is that its discriminant is non-positive,
    \begin{equation}
        \left(2\braket{\mathcal{P}_1}\braket{\mathcal{P}_2}\right)^2 - 4\left(1-\braket{\mathcal{P}_2}^2\right)\left(1-\braket{\mathcal{P}_1}^2\right) \leq 0.
    \end{equation}
    This straightforwardly leads to $\braket{\mathcal{P}_1}^2+\braket{\mathcal{P}_2}^2\leq1$.
\end{proof}
Note that this inequality can also be derived from the Cauchy--Schwarz inequality between $\mathcal{P}_1-\braket{\mathcal{P}_1}$ and $\mathcal{P}_2-\braket{\mathcal{P}_2}$, with the inner product defined by $A\cdot B=\braket{\{A,B\}}/2$. From the lemma, we have
\begin{corollary}\label{co1}
	If $\braket{U_X}\geq 1-\epsilon$, then $|\braket{Z_1X_2}|$, $|\braket{X_{N-1}Z_N}|$, $|\braket{Z_iZ_{i+1}X_{i+2}}|$,  and $|\braket{X_iZ_{i+1}Z_{i+2}}|\leq\sqrt{2\epsilon}$ for $1\leq i\leq N-2$.
\end{corollary}\noindent
With this result, we are able to efficiently estimate all the stabilizers.

The second lemma asserts that estimating stabilizers is enough to certify the graph state.
\begin{lemma}\label{lm2}
	Denote $S_g$ as the stabilizer group of a stabilizer state $\ket{\psi_g}$. Then the following relations hold:
	\begin{itemize}
		\item Let $\{\mathcal{P}_i\}$ be a set of generators of $S_g$. If for all $\mathcal{P}_i$, $\braket{\mathcal{P}_i}\geq 1-\epsilon_i$, then $\bra{\psi_g}\rho\ket{\psi_g}\geq 1-\sum_i\epsilon_i/2$.
		\item If $\fd\geq 1-\epsilon$, then $\braket{\mathcal{P}}\geq1-2\sqrt{\epsilon}$ for all $\mathcal{P}\in S_g$.
	\end{itemize}
\end{lemma}
\begin{proof}
The state $\ket{\psi_g}$ is the unique ground state of the Hamiltonian $H=\sum_{i\in S_g} (1-\mathcal{P}_i)$, whose gap is 2.
When $\braket{\mathcal{P}_i}\geq 1-\epsilon_i$ holds for all $i$,
we have 
\begin{align}
	\braket{H}=n-\sum_i\braket{\mathcal{P}_i}\leq \sum_i\epsilon_i.
\end{align}
From the fact that $H\succeq 2(1-\ket{\psi_g}\bra{\psi_g})$, we have $\braket{H}\geq 2(1-\bra{\psi_g}\rho\ket{\psi_g})$. As a result,
\begin{align}
	\fd\geq 1-\sum_i\epsilon_i/2.
\end{align}

On the other hand, when
 $\fd\geq 1-\epsilon$,  the Fuchs-van de Graaf inequality relates fidelity to 1-norm $\Vert\cdot\Vert_1$(see \cite{Nielsen_Chuang_2010})
\begin{align}
	\Vert {\rho}-\ket{\psi_g}\bra{\psi_g}\Vert_1\leq \sqrt{\epsilon}.
\end{align}
So for projector $(1-\mathcal{P}_i)/2$, 
\begin{align}
		\text{Tr}\big(({\rho}-\ket{\psi_g}\bra{\psi_g})(1-\mathcal{P}_i)/2\big)
		\leq 
		\Vert {\rho}-\ket{\psi_g}\bra{\psi_g}\Vert_1\leq \sqrt{\epsilon}.
\end{align}
Using $\text{Tr}\big(({\rho}-\ket{\psi_g}\bra{\psi_g})\mathcal{P}_i\big)=\braket{\mathcal{P}_i}-1$, we have
\begin{align}
	\braket{\mathcal{P}_i}\geq 1-2\sqrt{\epsilon}.
\end{align}
\end{proof}

We note that above lemmas together give an upper bound on the infidelity.
\begin{corollary}
Let $\ket{\psi_g}$ be the 1D graph state defined in the main text.
    The infidelity with any lab state $\rho$ satisfies
    \begin{align}
        1-\bra{\psi_g}\rho\ket{\psi_g}\leq
        \frac{1}{2}\left(2(N-1)\sqrt{1-\braket{U_X}^2}-\sum_{i=1}^N(1-\braket{M_i})\right)
    \end{align}
\end{corollary}
\begin{proof}
    By Lemma~\ref{lm1}, we have that for $2\leq i\leq N-1$,
    \begin{align}
        |\braket{ZZX}|\leq \sqrt{1-\braket{U_X}^2},\quad
        |\braket{XZZ}|\leq \sqrt{1-\braket{U_X}^2},
    \end{align}
    and for $i=1,N$,
    \begin{align}
        |\braket{ZX}|\leq \sqrt{1-\braket{U_X}^2},\quad
        |\braket{XZ}|\leq \sqrt{1-\braket{U_X}^2}.
    \end{align}

    By Lemma~\ref{lm2}, we can bound the infidelity,
    \begin{align}
        1-\bra{\psi_g}\rho\ket{\psi_g} &\leq \frac{1}{2}\left(N-\sum_{i=1}^N\braket{\mathcal{P}_i}\right)\notag\\
        &\leq  \frac{1}{2}\Bigg(N-\sum_{i=2}^{N-1}\left(\braket{M_i}-2\sqrt{1-\braket{U_X}}\right)
        \notag\\
        &\quad\quad -\sum_{i=1,N}\left(\braket{M_i}-\sqrt{1-\braket{U_X}}\right)\Bigg)
        \notag\\
        &=\frac{1}{2}\left(2(N-1)\sqrt{1-\braket{U_X}^2}-\sum_{i=1}^N(1-\braket{M_i})\right)
    \end{align}
\end{proof}

The algorithm needs to estimate $\braket{U_X}$, so we also need to bound the fidelity by $\braket{U_X}$.
\begin{lemma}\label{lm3}
	If $\braket{U_X}< 1-\epsilon$, then $\fd<1-\epsilon/2$.
\end{lemma}
\begin{proof}
Let $P_\pm$ be the projector to the $\pm$ subspaces of $U_X$, respectively.
$U_X=P_+-P_-=2P_+-1$. If $\braket{U_X}<1-\epsilon$, $\braket{P_+}<1-\epsilon/2$. From $\ket{\psi_g}\bra{\psi_g}\prec P_{+}$, $\fd\leq\braket{P_{+}}<1-\epsilon/2$.
\end{proof}

With all the technical preparations, we now prove our main theorem.
\begin{theorem}
	For any $\epsilon<1/(64N^2)$, Algorithm 1 outputs \textsc{Certified} with probability at least $2/3$ when $\fd>1-\epsilon$, and outputs \textsc{Failed} with probability at least $2/3$ when $\fd<1-8N\sqrt{\epsilon}$, using a number of samples at most $T=\mathcal{O}\big(1/\epsilon^2\big)$.
\end{theorem}
\begin{proof}
The algorithm works by estimating $\braket{U_X}$ and $\braket{M_i}$ for $i=1,2,\cdots,N$, using empirical averages $\emp{U_X}$ and $\emp{M_i}$.  
If in a single execution of the algorithm, the estimation satisfies  the following error bound 
\begin{align}
    |\emp{U_X}-\braket{U_X}|&\leq \frac{5}{4}\epsilon, \notag\\
    |\emp{M_i}-\braket{M_i}|&\leq\sqrt{\epsilon},
\end{align}
 we regard it as a successful execution. 

We first show that in successful executions, the algorithm outputs the correct answer.
In the case of high fidelity, assume $\bra{\psi_g}{\rho}\ket{\psi_g}>1-\epsilon$. If Step 2 outputs \textsc{Failed}, $\emp{U_X}<1-13/4\epsilon$, so $\braket{U_X}<1-2\epsilon$. Then according to Lemma \ref{lm3}, $\bra{\psi}\rho\ket{\psi}<1-\epsilon$, contradicting the assumption. So Step 2 never outputs \textsc{Failed}. As a result, $\emp{U_X}\geq1-13/4\epsilon$ and $\braket{U_X}\geq 1-4\epsilon$. Then we have $|\braket{Z_1X_2}|$, $|\braket{X_{N-1}Z_N}|$, $|\braket{Z_iZ_{i+1}X_{i+2}}$,  and $|\braket{X_iZ_{i+1}Z_{i+2}}|\leq3\sqrt{\epsilon}$ according to Corollary \ref{co1}, for any $i=1,2,\cdots, N-2$.

According to Lemma \ref{lm2}, for the high-fidelity case, we have $\braket{X_1Z_2}$, $\braket{Z_{N-1}X_N}$, and $\braket{Z_iX_{i+1}Z_{i+2}}>1-2\sqrt{\epsilon}$ for any $i=1,2,\cdots, N-2$. So $\braket{M_i}>1-8\sqrt{\epsilon}$ for all $i$.  As a result, $\emp{M_i}>1-9\sqrt{\epsilon}$ always holds. As a result, the algorithm outputs \textsc{Certified}. This verifies the correct output in the high-fidelity case.

Then we deal with the low-fidelity case. Assume $\bra{\psi}\rho\ket{\psi}<1-8N\sqrt{\epsilon}$ and Step 2 does not correctly output \textsc{Failed}.
 According to Lemma \ref{lm2}, there must exist at least one stabilizer generator $\mathcal{P}_k$ such that $\braket{\mathcal{P}_k}<1-16\sqrt{\epsilon}$. Then we have $\braket{M_k}<1-10\sqrt{\epsilon}$. As a result, $\emp{M_k}<1-9\sqrt{\epsilon}$ always holds for this $k$, and the algorithm outputs \textsc{Failed}. This verifies the correct output in the low-fidelity case.

 The above arguments confirm that the algorithm outputs the correct result in successful executions. In the following, we bound the probability of successful executions.

Using Hoeffding's inequality, from $\braket{U_X}\in[-1,1]$ we have 
\begin{align}\label{prob1}
\text{Pr}\left(|\emp{U_X}-\braket{U_X}|>\frac{5}{4}\epsilon\right)\leq 2\exp\left(-\frac{25T\epsilon^2}{32}\right).
\end{align}
With $\braket{M_i}\in[-3,3]$, we have
\begin{align}\label{prob2}
\text{Pr}\left(|\emp{M_i}-\braket{M_i}|>\sqrt{\epsilon}\right)\leq 2\exp\left(-\frac{T\epsilon}{18}\right),
\end{align}
When 
\begin{align}
    T\geq \frac{32\log12}{25\epsilon^2},
\end{align}
we have
\begin{align}
    2\exp\left(-\frac{25T\epsilon^2}{32}\right)&\leq  \frac{1}{6};\notag\\
    2\exp\left(-\frac{T\epsilon}{18}\right)&\leq 2\exp\left(-\frac{16\log12}{225\epsilon}\right)\notag\\
    &\leq 2\exp\left(-\frac{1024\log12}{225}N^2\right).
\end{align}
Here we use $\epsilon\leq1/(64N^2)$. So by the union bound, the failure probability is 
\begin{align}
    1-\operatorname{Pr}_{\operatorname{succ}}&\leq \frac{1}{6}+2N\exp\left(-\frac{1024\log12}{225}N^2\right)\notag\\
    &\leq \frac{1}{6}+2\exp\left(-\frac{1024\log12}{225}\right)\notag\\
    &<\frac{1}{3}.
\end{align}
This verifies the $2/3$ success probability using $T=\mathcal{O}(1/\epsilon^2)$ samples.

\end{proof}

\section{Two conjectures}
For the 1d graph state with $N$ even, the same protocol no longer works. However, provided that the following conjecture is correct, we have devised a more involved protocol to identify the state. 

\textbf{Conjecture 1.} If a state can be uniformly certified under the condition that the state is a stabilizer state, the condition can be removed.

\textbf{Conjecture 2.} 
If a state can be uniformly certified under the condition that the state is pure, the condition can be removed.

The two conjectures stem from our numerical explorations of uniform certification. The intuition is that, although counting-wise uniform measurements do not seem very restrictive on a general $2^N\times2^N$ matrix, the nonlinear and highly nontrivial condition of positive-semi-definiteness exerts another layer of restriction. Conjecture 2 is weaker and, if proved, could provide confidence in Conjecture 1.

\section{Uniform measurements for stabilizer states}

In this section we consider a modified problem in which the given state is \textit{promised} to be a stabilizer state.
A celebrated known example is \textit{tomography} of stabilizer states.
Given a resource that can produce a state $\ket{\psi}$ which is promised to be a stabilizer state,
it can be proved that $O(n)$ samplings from a resource is enough for extracting full knowledge of $\ket{\psi}$ \cite{montanaroLearningStabilizerStates2017}. 
However, the protocol in \cite{montanaroLearningStabilizerStates2017} depends on
measuring in the two-qubit Bell basis and 
reading out the outcomes from single measurements rather than some expectation value,
making the protocol difficult to realize by uniform measurements.

Here we discuss \textit{certification} of a state $\ket{\psi}$
which is promised to be a stabilizer state but it is unknown whether it will be the target state $\ket{\psi}_{\text{targ}}$.
Since the set of stabilizer states is not continuous,
it is not reasonable to discuss fidelity here.
This may seem strange,
but it can be interesting, since we will see that uniform measurements can have surprising power even when $\ket{\psi}_{\text{targ}}$ is not of CSS type.

As a non-trivial example, let the target state be a 1D graph state on an even number of qubits.
The stabilizer group of such a state is generated by $X_1Z_2,Z_{N-1}X_N$, and $Z_{i-1}X_{i}Z_{i+1}$ for $2\leq i \leq N-1$.
When $N$ is even, we no longer have a stabilizer comprising only one type of Pauli operator.
To certify this state, we still measure uniformly in $X$, $Z$, $X+Z$, $X-Z$ directions, respectively.
If the state is in $\ket{\psi}$, we will get 
\begin{align}
    1 &= \braket{Z_{1}X_{2}} + \braket{X_{1}Z_{2}}, \label{eq:cl-cond-1} \\
    1 &= \braket{Z_{N-1}X_{N}} + \braket{X_{N-1}Z_{N}}, \label{eq:cl-cond-2} \\
    1 &= \braket{X_{i-1}Z_i Z_{i+1}} + \braket{Z_{i-1}X_i Z_{i+1}} + \braket{Z_{i-1}Z_i X_{i+1}}, \label{eq:cl-cond-3} \\
    0 &=\braket{X_{i}X_{i+3}}. \label{eq:cl-cond-4}
\end{align}
For each equation above, $i$ takes the value such that all indices in the equation are valid.
We now show that these conditions are sufficient for $\ket{\psi}$ to be the graph state
given the promise that it is a stabilizer state.

\begin{figure*}
    \centering
    \includegraphics[width = \linewidth]{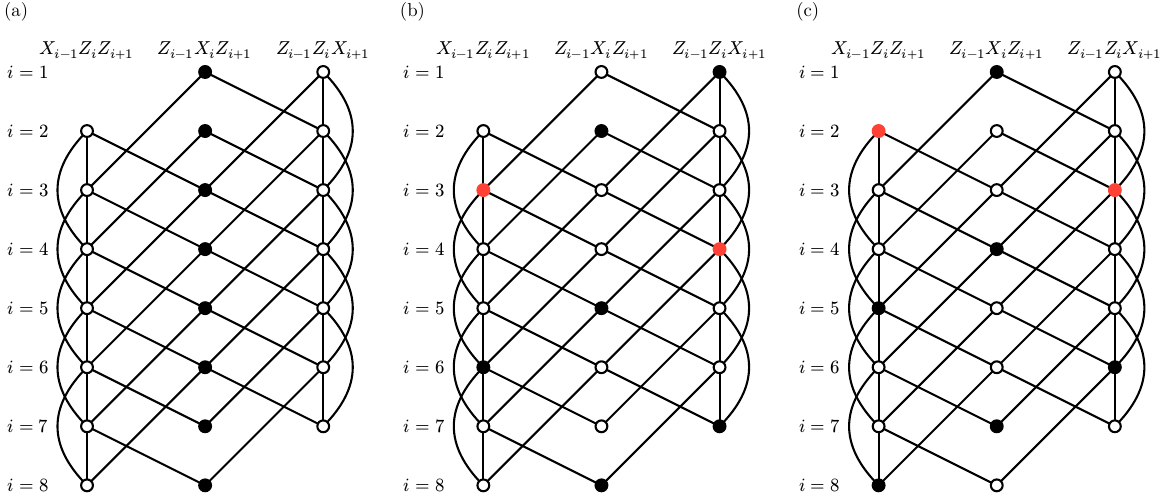}
    \caption{The only three possibilities satisfying conditions (\ref{eq:cl-cond-1}, \ref{eq:cl-cond-2}, \ref{eq:cl-cond-3}), with $N = 8$ number of qubits.
    Single Pauli strings are represented by circles,
    with anti-commuting operators connected by a link.
    For the first line, $i-1$ is an invalid index,
    and thus the two circles represent $X_1 Z_2$ and $X_2 Z_1$.
    Similarly, two circles in the last line represent $X_{N-1}Z_N$ and $Z_{N-1}X_N$.
    Solid circle indicates that the operator has expectation value 1
    and empty circle indicates expectation value 0.
    Due to our analysis in the main text, 
    the last two violate condition (\ref{eq:cl-cond-4}) (represented by red circles),
    making graph state the only possibility compatible with measurement results.
    }
    \label{fig:stab-cls}
\end{figure*}

For a stabilizer state, the expectation value of a Pauli string can only be $\pm 1$ or 0.
Moreover, if a Pauli string $P$ has expectation value $\pm 1$,
then any Pauli string that anti-commutes with $P$ has expectation value 0 according to Lemma 1.
As in Fig.~\ref{fig:stab-cls}, we can plot Pauli strings appearing in Eqs.~(\ref{eq:cl-cond-1}-\ref{eq:cl-cond-3}) as circles,
such that two Pauli strings in Eq.~(\ref{eq:cl-cond-1}) are put in the first line,
two Pauli strings in Eq.~(\ref{eq:cl-cond-2}) in the last line,
while the three-body Pauli strings in Eq.~(\ref{eq:cl-cond-3}) are put such that 
each line contains 3 Pauli strings with the same $i$.
The anti-commuting Pauli strings are connected by a link.
We now need to assign expectation values to each Pauli string in the figure.

As for the first line, we must assign a 1 and a 0 since they sum up to 1.
If we assign 1 to $X_1 Z_2$ (indicated by the solid circle), then $Z_1Z_2X_3$ must be assigned 0,
forcing the first two terms in the second line to be one 1 and one 0.
The different choice of how to assign 1 to the second line leads to two possibilities,
shown in Fig.~\ref{fig:stab-cls}(a, c).
Then other assignments are automatically fixed.
If we assign 1 to $Z_1X_2$, all other assignments are fixed as in Fig.~\ref{fig:stab-cls}(b).
We see that now we have 3 possibilities.

However, in the last two possibilities, we see that we must have a stabilizer (represented by red circles in Fig.~\ref{fig:stab-cls}(b, c))
\begin{equation}
    (X_{i-1}Z_iZ_{i+1})(Z_i Z_{i+1} X_{i+2}) = X_{i-1}X_{i+2},
\end{equation}
which should be 0 according to condition (\ref{eq:cl-cond-4}).
We thus rule out these two possibilities and confirm that the state is the graph state.

\section{More examples of certifiable states}
In this section, we provide more examples of uniformly certifiable states, both within and beyond the bipartite construction given in the main text. We will also discuss a non-example to help develop intuitions on the capacity of uniform measurement.

In the main text, we perform a thorough analysis on the 1d graph state, the simplest realization of the bipartite graph construction. This is easily generalized to graph states a 2d square lattice, see Fig.~\ref{fig:2d}(a). The rough edges on the boundary of the graph guarantee that each filled vertex has four neighbors, and thus the graph belongs to the bipartite construction. Note that both the length and width of the lattice have to be odd.

Perhaps surprisingly, graph states defined on a smooth-boundary square lattice with odd-length sides can also be uniformly certified. To see this, we partition the square lattice into three sublattices, $A, B, C$, see Fig.~\ref{fig:2d}(b). The key observation is that (i) each vertex in $B$ has even neighbors in $A$, and (ii) all vertices in $A$ and $C$ are only connected to $B$. Similar to the bipartite construction, the conditions ensure that $\prod_{a\in A}X_a$ commutes with all the stabilizers, which kills all the unwanted terms from uniform measurements hinged on $a\in A$. This result then kills the unwanted terms hinged on $b\in B$ and $c\in C$ in turn, and completes the uniform certification procedure. It is not hard to see that the condition (ii) can be generalized to a $n$-partite chain $A-B-C-...$.

Finally, we discuss a non-example where uniform measurement fails to resolve two different stabilizer states. Consider the two states $XZ|\psi\rangle=-ZX|\psi\rangle=|\psi\rangle$ and $-XZ|\phi\rangle=ZX|\phi\rangle=|\phi\rangle$. Then, by explicitly checking all symmetric operators on the two qubits, we can see that both states always give the same measurement outcomes, and thus indistinguishable with uniform measurements alone. Indeed, all states $a|\psi\rangle\langle\psi|+(1-a)|\phi\rangle\langle\phi|, 0\leq a\leq1$ are indistinguishable. This is a minimal example where local measurements are required to certify a quantum state.

\begin{figure}
    \centering
    \includegraphics{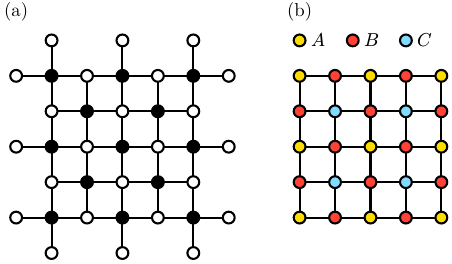}
    \caption{(a) A rough-boundary 2d lattice that belongs to the bipartite construction given in the main text. (b) A smooth-boundary 2d lattice can also be uniformly certified by dividing into a tripartite structure.}
    \label{fig:2d}
\end{figure}

\end{document}